\documentclass[
reprint,
superscriptaddress,
amsmath,amssymb,
aps,
pra,
]{revtex4-2}
\usepackage{multirow}
\usepackage{booktabs}
\usepackage{floatflt}
\usepackage{ragged2e}

\usepackage{floatrow}
\usepackage{graphicx}
\usepackage[label font=bf,labelformat=simple]{subfig}
\usepackage{filecontents}
\usepackage{pgfplots}
\usepackage{amsmath}
\usepackage{amssymb}
\usepackage{tikz}
\usepackage{dcolumn}
\usepackage{bm}
\usepackage{blindtext}
\usepackage{listings}
\usepackage{physics}
\usepackage{stmaryrd}
\usepackage[colorlinks,linkcolor=blue,anchorcolor=blue,citecolor=blue]{hyperref}
\makeatletter
\newif\if@restonecol
\makeatother

\usepackage[ruled,linesnumbered,vlined]{algorithm2e}
\usepackage{algpseudocode}

\usepackage{amsthm}
\usepackage{graphicx}
\usepackage{dcolumn}
\usepackage{bm}
\usepackage[utf8]{inputenc}

\begin{document}

\hypersetup{breaklinks=true}

\preprint{APS/123-QED}

\title{Unveiling NPT bound problem: From Distillability Sets to Inequalities and Multivariable Insights}

\author{Si-Yuan Qi}
\affiliation{CAS Key Laboratory of Quantum Information, University of Science and Technology of China, Hefei, Anhui 230026, P. R. China}
\affiliation{CAS Center For Excellence in Quantum Information and Quantum Physics, University of Science and Technology of China, Hefei, Anhui 230026, P. R. China}

\author{Geni Gupur}
 \email{genigupur@vip.163.com}
\affiliation{College of Mathematics and Systems Science, Xinjiang University, Urumqi, Xinjiang 830017, People's Republic of China}

\author{Yu-Chun Wu}
 \email{wuyuchun@ustc.edu.cn}
\affiliation{CAS Key Laboratory of Quantum Information, University of Science and Technology of China, Hefei, Anhui 230026, P. R. China}
\affiliation{CAS Center For Excellence in Quantum Information and Quantum Physics, University of Science and Technology of China, Hefei, Anhui 230026, P. R. China}
\affiliation{Institute of Artificial Intelligence, Hefei Comprehensive National Science Center, Hefei, Anhui 230088, People’s Republic of China}

\author{Guo-Ping Guo}
\affiliation{CAS Key Laboratory of Quantum Information, University of Science and Technology of China, Hefei, Anhui 230026, P. R. China}
\affiliation{CAS Center For Excellence in Quantum Information and Quantum Physics, University of Science and Technology of China, Hefei, Anhui 230026, P. R. China}
\affiliation{Institute of Artificial Intelligence, Hefei Comprehensive National Science Center, Hefei, Anhui 230088, People’s Republic of China}
\affiliation{Origin Quantum Computing Hefei, Anhui 230026, P. R. China}

\date{\today}

\newtheorem{theorem}{Theorem}[section]
\newtheorem{lemma}[theorem]{Lemma}
\newtheorem{proposition}[theorem]{Proposition}

\theoremstyle{definition}
\newtheorem{definition}[theorem]{Definition}

\newcommand{\RomanNumeralCaps}[1]
    {\MakeUppercase{\romannumeral #1}}

\setcounter{MaxMatrixCols}{20}

\begin{abstract}
Equivalence between Positive Partial Transpose (PPT) entanglement and bound entanglement is a long-standing 
open problem in quantum information theory.
So far limited progress has been made, even on the seemingly simple case of Werner states bound entanglement.
The primary challenge is to give a concise mathematical representation of undistillability. 
To this end, we propose a decomposition of the N-undistillability verification into $log(N)$ repeated steps of 1-undistillability verification.
For Werner state N-undistillability verification, 
a bound for N-undistillability is given, which is independent of the dimensionality of Werner states. 
Equivalent forms of inequalities
for both rank one and two matrices are presented, before transforming the two-undistillability
case into a matrix analysis problem. 
A new perspective is also attempted by seeing it as a non-convex multi-variable function,
proving its critical points and conjecturing Hessian positivity, which would make them local minimums.
\end{abstract}

\maketitle

\section{Introduction}

Entanglement \cite{RevModPhys.81.865} is what distinguished the quantum world from classical worlds. 
Unfortunately, our understanding of it remains quite limited to this day.
Entanglement exhibits varying degrees of strength, the highest of which falls into a subclass called 
``maximal entanglement''. Bell state is one such example. No consensus has yet been reached about what 
qualifies as a ``weak entanglement''. Quantum states with small Schmidt numbers are generally seen as weakly entangled, 
since quantum states with Schmidt number 1 are separable. 
Undistillable entanglement is another kind of widely accepted
weak entanglement. Relationships between these two kinds of weak entanglements are discussed in 
\cite{Huber_2018}.

Entanglement distillation \cite{Bennett_1996} is the process of 
obtaining pure Bell states only by Local Operations and Classical Communication (LOCC) 
from multiple copies of an entangled bipartite state. If N copies of a given 
entangled bipartite state can be transformed into pure Bell states through 
LOCC, then the entangled bipartite state is referred to as N-distillable. If not, then it's N-undistillable. 
If a particular entangled bipartite state is N-undistillable for arbitrary N, 
it is considered undistillable and is also referred to as a bound entangled state.

Another property concerned here is ``Positive Partial Transpose'' (PPT). PPT refers to a bipartite 
quantum state whose
partial transpose is positive semidefinite.
It is established in \cite{9801069} that Positive Partial Transpose (PPT) entanglement implies bound entanglement,
and therefore weakly entangled in terms of distillability.
However, the question of whether the converse is true remains an open problem \cite{horodecki2020open}. 
This problem is usually referred to as the NPT bound problem, or the distillability problem.
The popular conjecture regarding this problem is that the converse is not true, 
and that NPT bound entanglement exists. A number of special circumstances
have been worked out in \cite{Chen_2016_2, Chen_2023dam, Chen_2021, qian2019matrix, ding2023entanglement}.
Indirect approaches to this problem include expansion of distillation operations to 
k-extendible operations \cite{pankowski2013entanglement}, catalysis-assisted distillation \cite{lami2023catalysis},
and dually non-entangling and PPT-preserving channels \cite{Chitambar_2020}, linking
the problem with squashed entanglement \cite{Brand_o_2011} and linear preserver \cite{Johnston_2011},
studying the transformation of bound entangled states under certain dynamic process 
\cite{sharma2016dzyaloshinskiimoriya}, or considering the problem
in a broader scenario such as  
hyperquantum states \cite{van_der_Eyden_2022}.

In three separate attempts to directly tackle the distillability problem \cite{9910026}, \cite{9910022} and \cite{Chen_2016}, 
they each had a subset of quantum
states singled out,
and proved that for any finite number of copies N, there exist undistillable states in this subset. However, 
as N approaches infinity, the subset shrinks to emptiness. It has been shown in \cite{0312123} that N-undistillability does not
imply (N+1)-undistillability, by demonstrating a set of distillable-only-by-arbitrarily-large-number-of-copies state,
thus ruling out the easy route of concentrating on few-copies distillation.

For the Werner states \cite{PhysRevA.40.4277}, a family of single-parameter states,
the existence of NPT bound states implies existence of NPT bound Werner states \cite{9708015}.
Therefore, it is sufficient to only study Werner states.
Due to the inherent complexity of addressing N-copy Werner state distillability, over the last decade,
direct attacks on the problem have been primarily
fixating on the seemingly straightforward problem of 2-copy Werner state distillability, but so far only 
limited progress has been made.
In an attempt to attack the 2-copy Werner state distillability problem,
the special case of $4\times 4$ bipartite states is considered in \cite{5508622}, in which case the distillability
condition can be reduced to a convenient form of ``half-property'', and subsequently reformulated the distillability
problem into a matrix analysis problem. For a special case where the matrices are normal, a proof was given.

In \cite{1003.4337} the problem is reformulated and established that for a $d\times d$ bipartite Werner state,
the 2-copy Werner state undistillability problem
is equivalent to a certain $2d^2 \times 2d^2$ matrix being positive semidefinite. Positive definiteness
of the block matrices in the upper-left and lower-right corners has been proved.

In this work, we introduce a method that decomposes bound entanglement verification into repeated steps of 
1-undistillability verification, before
giving a new bound for N-undistillability that is independent of quantum state dimensions $d$, and therefore
different from that given in \cite{9910026}, \cite{9910022} and \cite{Chen_2016}. However, similar
to the problem encountered in \cite{9910026}, \cite{9910022} and \cite{Chen_2016}, our set of N-undistillable states 
also shrinks to emptiness for $N\to \infty$.
Also, an equivalent matrix analysis inequality form is presented, which is applicable to all finite
dimensionalities, as opposed to the case
in \cite{5508622} when only the $d=4$ case is considered. Another equivalent version of inequalities regarding only
rank-one matrices is also included, which takes on a form reminiscent of Cauchy-Schwartz inequality.
Finally, we take on a brand-new perspective of considering it as a multi-variable function problem,
finding critical points, proving non-convexity and conjecturing Hessian positivity.

We will be using $\mathcal{H}_A$ and $\mathcal{H}_B$ to denote Hilbert spaces pertaining
to A and B subsystems, and $X_i$ to denote matrices of rank less than or equal to $i$. 
Throughout the paper we 
will be considering multi-partite quantum states with identical
dimensions, i.e. quantum states on $\mathcal{H}_A\otimes \mathcal{H}_B$, where $dim(\mathcal{H}_A) = dim(\mathcal{H}_B)=d$.
This is justified by noticing that any quantum state of $d_1\times d_2$ can be equivalently transformed
into a quantum state of identical dimensions $max\{d_1, d_2\}\times max\{d_1, d_2\}$, with all unnecessary
elements set to zero.

The structure of this paper is as follows:
In Section \RomanNumeralCaps{2}, we 
introduce a matrix transformation that takes 
N-distillability problem into a higher dimensional 1-distillability problem.
In Section \RomanNumeralCaps{3}, the process of $N$-undistillability verification is transformed
into $log(N)$ repeated steps of 1-undistillability verification.
In Section \RomanNumeralCaps{4}, upper bounds on $N$-undistillability are given, 
yielding a result similar to that of \cite{9910026}, \cite{9910022} and \cite{Chen_2016},
effectively finding shrinking families
of N-undistillable quantum states.
In Section \RomanNumeralCaps{5}, equivalent partial trace inequalities are presented before turning them into a matrix
analysis problem.
The rank-two matrix analysis problem is also reduced to a 
rank-one matrix analysis problem, resulting in a Cauchy-Schwartz inequality look-alike.
In Section \RomanNumeralCaps{6}, we take on a new perspective by treating the problem as 
a multi-variable function, proving critical points, non-convexity and conjecturing 
positive-semidefinite Hessians.

\section{Statement of the distillability problem}
The open problem of equivalence between PPT entanglement and bound entanglement can be resolved by 
finding bound entangled NPT states.
Since Peres \cite{9604005} has already shown that NPT is a sufficient condition for entanglement, attention 
can be limited to the state being both ``bound'', which is, N-undistillable for arbitrarily large N, and
NPT, which is, having a negative partial transpose. 
Compared to ascertaining positive semidefiniteness
of a certain matrix, the major difficulty usually lies in finding a concise mathematical
interpretation for undistillability, and asserting it for two types of 
arbitrariness: arbitrary LOCC operation and arbitrarily finitely many copies of the entangled state. 
This will be the main focus of this section,
and the analysis of mathematical structure will eventually lead to a transformation that reduces 
N-undistillability verification to 1-undistillability verification of a transformed state.

Let's recall the concept of Schmidt decomposition, Schmidt coefficients and Schmidt rank \cite{Nielsen_Chuang_2010}.
Schmidt decomposition of a pure state $|\psi\rangle$ in 
$\mathcal{H}_A\otimes \mathcal{H}_B$ is:

\begin{equation}
    |\psi\rangle = \sum_{i=1}^{d} \alpha_i |u_i v_i\rangle,
\end{equation}
where $\mathcal{H}_A$ and $\mathcal{H}_B$ are both d-dimensional Hilbert spaces, with 
$\{|u_1\rangle, \cdots, |u_d\rangle\}$ and $\{|v_1\rangle, \cdots, |v_d\rangle\}$ as their respective
orthonormal bases. Schmidt coefficients $\alpha_i$ are real, nonnegative and unique up to re-ordering.
The number of non-zero Schmidt coefficients is called Schmidt rank of pure state $|\psi\rangle$. 
$SR(|\psi\rangle)$ denotes the Schmidt rank of $|\psi\rangle$.

A bipartite quantum state is called N-distillable if pure Bell states can be obtained using 
LOCC, from N-copies of the state. 
It has been proven in \cite{PhysRevLett.78.574} that all entangled $2\times 2$ bipartite quantum states are distillable,
therefore, all we need to do is transform the original quantum state into a $2\times 2$ entangled state by LOCC.
This is demonstrated in the following result:
A bipartite quantum state $\rho$ on 
$\mathcal{H}_A\otimes \mathcal{H}_B$ is N-undistillable
iff for any $d^N\times d^N$ bipartite
pure state $|\psi^{SR\leq 2}\rangle$ on $\mathcal{H}_A^{\otimes N}\otimes \mathcal{H}_B^{\otimes N}$ 
with Schmidt rank $\leq$ 2, the following holds:
\begin{equation}\label{gen2}
    \langle \psi^{SR\leq 2} | (\rho^{T_A})^{\otimes N} |\psi^{SR\leq2}\rangle \geq 0,
\end{equation}
where $\rho^{T_A} = (T\otimes I)\rho$ is the partial transpose of $\rho$, with $T$ being the transpose operator.
The above expression overlooks the pairing between different subsystems and so an additional operator 
$M_N$ from $(\mathcal{H}_A\otimes\mathcal{H}_B)^{\otimes N}$ 
to $\mathcal{H}_A^{\otimes N}\otimes \mathcal{H}_B^{\otimes N}$ is introduced to make it meaningful:

\begin{definition}
    Operator $M_N: (\mathcal{H}_A\otimes\mathcal{H}_B)^{\otimes N} 
    \to \mathcal{H}_A^{\otimes N}\otimes \mathcal{H}_B^{\otimes N}$ is
    defined as:
    \begin{equation}
        M_N = \sum
        |i_{A_1}\dots i_{A_N}i_{B_1}\dots i_{B_N}\rangle
        \langle i_{A_1}i_{B_1}\dots i_{A_N}i_{B_N}|,
    \end{equation}
    where $\mathcal{H}_A^{\otimes N} = \mathcal{H}_{A_1}\otimes \mathcal{H}_{A_2}\otimes \cdots\otimes \mathcal{H}_{A_N}$
, $\mathcal{H}_B^{\otimes N} = \mathcal{H}_{B_1}\otimes \mathcal{H}_{B_2}\otimes \cdots\otimes \mathcal{H}_{B_N}$, 
and $i_{A_j}, i_{B_j}$ stands for the index corresponding to system $A_j, B_j$. 
\end{definition}

Operator $M_N$ then takes the $A_1 B_1 A_2 B_2 \cdots A_N B_N$ structure of $(\rho^{T_A})^{\otimes N}$
and merges their A-subsystems and B-subsystems,
resulting in a structure of $A_1 A_2 \cdots A_N B_1 B_2 \cdots B_N$ that suits the pure quantum state,
thus explicitly demonstrating the process of system pairing which is essential to any further calculations
derived from this representation of undistillability.

For better representation of undistillability and further use in following sections,
we use a certain state-operator isomorphism introduced in \cite{5508622}.
We notice that it is in essence a Choi-Jamiokowski isomorphism followed by a transposition on the resulting matrix.
For convenience and consistency with previous works, we stick to this ``transposed Choi-Jamiokowski isomorphism''
when presenting our findings.

\begin{definition}[State-operator isomorphism \cite{5508622}]\label{iso}
    A state-operator isomorphism $\Psi$ is one that takes
    a pure bipartite state $|\psi\rangle$ of:
\begin{equation}
    |\psi\rangle = \sum_{ij}\alpha_{ij}|ij\rangle,
\end{equation}
to a matrix form of:
\begin{equation}
    \Psi(|\psi\rangle) = \begin{pmatrix}
        \alpha_{00} & \alpha_{01} &  \cdots & \alpha_{0, d-1}\\
        \alpha_{10} & \alpha_{11} &  \cdots & \alpha_{1, d-1}\\
        \vdots & \vdots & \ddots & \vdots \\
        \alpha_{d-1, 0} & \alpha_{d-1, 1} &  \cdots & \alpha_{d-1, d-1}\\
    \end{pmatrix}.
\end{equation}
\end{definition}

Note that Schmidt decomposition of a pure quantum state $|\psi\rangle$ is related to the singular value decomposition
of $\Psi(|\psi\rangle)$ through state-operator isomorphism. In fact, Schmidt coefficients of $|\psi\rangle$
is the same as the singular values of $\Psi(|\psi\rangle)$, and Schmidt rank of $|\psi\rangle$ is equal to 
the rank of $\Psi(|\psi\rangle)$.

The following is a more concise mathematical representation of N-undistillability:
A bipartite quantum state $\rho$ is N-undistillable iff:
\begin{equation}\label{maineq}
    [\Psi^{-1}(X_2)]^{\dagger} M_N(\rho^{T_A})^{\otimes N}M_N^{\dagger} \Psi^{-1}(X_2) \geq 0,
\end{equation}
where $X_2$ is any $d^N\times d^N$ matrix of rank $\leq 2$.
Note that $\Psi^{-1}(X_2)$ is equivalent to $|\psi^{SR\leq 2}\rangle$, 
a pure bipartite quantum state with Schmidt rank no larger than 2.
For the case of $N=1$, 1-undistillability is:
\begin{equation}
    [\Psi^{-1}(X_2)]^{\dagger} \rho^{T_A} \Psi^{-1}(X_2) = 
    \langle \psi^{SR\leq 2}| \rho^{T_A} |\psi^{SR\leq 2}\rangle \geq 0.
\end{equation}

Effectively, the problem of $\rho$ 
being N-copy undistillable is equivalent to the problem of 
$M_N(\rho^{\otimes N})M_N^{\dagger}$ being 1-undistillable, but on a Hilbert space with higher
dimensions $\mathcal{H'}_A \otimes \mathcal{H'}_B$, $dim(\mathcal{H'}_A) = dim(\mathcal{H'}_B) = d^N$.

Since discussions of 1-copy distillability problem rarely involves specification of finite dimensional Hilbert space 
dimensionality, and most of the theorems and methods are applicable to arbitrary dimensions, this transformation
poses an advantage.
In the next section, a specific symmetry of $M_{2^N}$ transformations is utilized for decomposing the increasingly complicated 
structure of $M_{2^N}$ into identical steps repeated for $N$ times. In fact, verification of
$2^N$-undistillability is enough to assert bound entanglement.

\section{N-undistillability problem as repeated steps}

From the observation that (N+1)-undistillability implies N-undistillability, we can get the following result 
which relaxes the N-undistillability requirement for all N to only an infinite sequence of N diverging to infinity:
\begin{proposition}
    $\{g_i\}$ is a sequence of positive real numbers diverging to infinity.
    A bipartite entangled quantum state is $g_i$-undistillable for arbitrary i, iff the state is
    bound entangled.
\end{proposition}
In the 2-copy case, $M_2$ transformation is simply:
\begin{equation}\label{defe}
    E = M_2 = \sum
    |i_{A_1}i_{A_2}i_{B_1}i_{B_2} \rangle\langle i_{A_1}i_{B_1}i_{A_2}i_{B_2}|,
\end{equation}
which is merely an exchange of the second and third parts of $B_1$ and $A_2$.
From now on we will always refer to the 2-copy case of $M_2$ as $E$, emphasizing that it is 
only a simple exchange as opposed to complicated $M_N$ of larger copy numbers.

According to Eq.\ref{maineq},
the quantum state $\rho$ is 2-undistillable iff $E(\rho\otimes\rho)E^{\dagger}$ is 1-undistillable.

In the 4-copy case, $M_4$ can actually be decomposed into two steps:
\begin{equation}
    \begin{aligned}
    &A_1 B_1 A_2 B_2 A_3 B_3 A_4 B_4 \\
    \to &A_1 A_2 B_1 B_2 A_3 A_4 B_3 B_4 \\
    \to &A_1 A_2 A_3 A_4 B_1 B_2 B_3 B_4.
    \end{aligned}
\end{equation}

Step 1: do the 2-copy exchange on subsystems 1,2 ($A_1 B_1 A_2 B_2\to A_1 A_2 B_1 B_2$) 
and subsystems 3,4 ($A_3 B_3 A_4 B_4\to A_3 A_4 B_3 B_4$).
This step is no different than in the 2-copy case, except that it is done on two 
systems simultaneously. Group every neighbouring 2 subsystems(in terms of $A, B$),
so that the eight parts are now seen as 
four parts ($A_1 A_2, B_1 B_2, A_3 A_4, B_3 B_4$
to $ABAB$). In order to separate
this newly obtained four-part state from 2-copy Werner states, we denote this state $\rho(e)\otimes \rho(e)$,
where
$\rho(e) = E(\rho \otimes \rho)E^{\dagger}$.

Step 2: 
Do the 2-copy exchange $E$ again, obtaining $\rho(e^2) = E(\rho(e)\otimes \rho(e))E^{\dagger}$, where
$E$ is taken with regard to the ``merged'' parts. 
This step takes once-transformed Werner states
into twice-transformed Werner states. 

The above steps of repeating the $E$ transformation twice, is equivalent to doing the $M_4$ translations
directly. The original 
state $\rho$ is 4-copy undistillable iff $\rho(e^2)$ is one-copy undistillable.

Similarly, in the $2^N$-copy case, E is decomposed into N steps,
with each step taking state $\rho(e^{i-1})\otimes \rho(e^{i-1})$, 
producing a state $\rho(e^i) = E(\rho(e^{i-1})\otimes \rho(e^{i-1}))E^{\dagger}$, 
verifying its positivity on Schmidt rank $\leq 2$ states and 
effectively ascertaining $2^i$-copy undistillability,
before passing two copies of this state on, as a starting state for the next step.
The following theorem demonstrates the process.
\begin{theorem}
    Define $\rho(e^0)=\rho$ to be the initial state, and
    \begin{equation}
        \rho(e^k) = E(\rho(e^{k-1})\otimes \rho(e^{k-1}))E^{\dagger},
    \end{equation}
    then $\rho(e^k)$ being 1-undistillable is equivalent to $\rho$ being $2^k$-undistillable.
\end{theorem}

In other words, we are essentially searching for an infinite sequence of points connected by the operation
$E$, with the first point being $\rho = \rho(e^0)$, 
second point being $\rho(e^1) = E(\rho(e^0)\otimes\rho(e^0))E^{\dagger}$, and the $(i+2)$-th 
point being $\rho(e^{i+1}) = E(\rho(e^i)\otimes\rho(e^{i}))E^{\dagger}$.
If points of this infinite sequence always fall within the set of 1-undistillability, then we can
safely say that the starting point $\rho$ is N-undistillable for arbitrary N, or that it is a bound state.

\section{Bounds for Werner state N-undistillability}

\subsection{Werner state 1-undistillability and NPT condition}
It has been shown in \cite{9708015} that it is sufficient to only consider Werner states for NPT bound problem, that
if an NPT bound state exists, there must be an NPT bound state within the family of Werner states.
Necessary and sufficient conditions of Werner state 1-undistillability and partial transpose negativity 
are obtained in this section.

We first introduce a lemma that would assist further attempts concerning inner product with Schmidt rank 
$\leq 2$ pure states:
\begin{lemma}\label{fact}
Let $|\psi\rangle$ be an arbitrary pure quantum state and
    $s_j$ are the Schmidt coefficients of $|\psi\rangle$ arranged in descending order, then 
    \begin{equation}\label{facteq}
       \;\sum_{j=1}^k s_j^2= max_{|\phi^{SR\leq k}\rangle} |\langle \psi|\phi^{SR\leq k}\rangle |^2,
    \end{equation}
    where $|\phi^{SR\leq k}\rangle$ is an arbitrary quantum pure state with Schimidt rank being less than k, namely, $SR(|\phi^{SR\leq k}\rangle)\leq k$.
\end{lemma}
Proof of this lemma is provided in Appendix \ref{fact_proof}.

A Werner state is as follows: 
\begin{equation}
    \rho_w = \frac{1}{d^2 + \beta d}(I + \beta F), 
\end{equation}
where $F = \sum_{ij}|ij\rangle\langle ji|$, $I$ is the unnormalized identity matrix, d denote the dimensionality of both parts of the system, and $\beta$ is a parameter characterizing 
the ``portion'' of swap operator $F$, ranging in $-1 \leq \beta\leq 1$.

Its partial transpose is:
\begin{equation}
    \rho_w^{T_A} = \frac{1}{d^2 + \beta d}(I + \beta G),
\end{equation}
where $G = \sum_{ij}|ii\rangle\langle jj|$, which is in fact the unnormalized density matrix of a pure maximally entangle state $|\Phi\rangle$, namely,
\begin{equation}
    G = d|\Phi\rangle\langle \Phi|, \; |\Phi\rangle = \sum_{i}\frac{1}{\sqrt{d}}
    |ii\rangle.
\end{equation}

When $\beta >0$, $\rho_w^{T_A}$ is always positive semidefinite, hence we consider $\beta<0$ case only. Combining Eq.\ref{maineq} and Eq.\ref{facteq}, we obtain
\begin{equation}
    \begin{aligned}
    &\langle \psi^{SR\leq 2}|\rho_w^{T_A}|\psi^{SR\leq 2}\rangle \\
    &= \frac{1}{d^2+\beta d}(1+\beta d|\langle \Phi|\psi^{SR\leq 2}
    \rangle|^2)\stackrel{L.\ref{fact}}{\geq} \frac{1+2\beta}{d^2+\beta d}.
\end{aligned}
\end{equation}
A Werner state $\rho_w$ is 1-undistillable iff 
$\langle \psi^{SR\leq 2}|\rho_w^{T_A}|\psi^{SR\leq 2}\rangle \geq 0$ for any
$|\psi^{SR\leq 2}\rangle$ pure bipartite quantum state of Schmidt rank no larger than 2.
Therefore, Werner state is 1-undistillable iff $\beta\geq -\frac{1}{2}$.

Furthermore, for arbitrary quantum state $|\psi\rangle$, we have
\begin{equation}
    \langle \psi|\rho_w^{T_A}|\psi\rangle =\frac{1}{d^2+\beta d}(1+\beta d|\langle \Phi|\psi
    \rangle|^2)\geq \frac{1+\beta d}{d^2+\beta d},
\end{equation}
with the equality achieved when $|\psi\rangle = |\Phi\rangle$.
Therefore, Werner state is NPT iff $\beta < -\frac{1}{d}$. Any NPT 1-undistillable Werner state can therefore
only be found within the parameter range of $-\frac{1}{2}\leq \beta < -\frac{1}{d}$, provided
$d > 2$.

\subsection{Sufficient conditions for Werner state N-undistillability}

We derive new bounds for Werner state N-undistillability in this section. 
As a result, for any finite N, a set of N-undistillable states can be found, 
but similar to the 
circumstances encountered in \cite{9910026}, \cite{9910022} and \cite{Chen_2016}, 
the set of N-undistillable states shrinks to emptiness as N approaches infinity.

In the previous section it has been established that
Werner state N-undistillability is equivalent to:
\begin{equation}
    \langle \psi^{SR\leq 2}|M_N(\rho_w^{T_A})^{\otimes N}M_N^{\dagger}|\psi^{SR\leq 2}\rangle\geq 0,
\end{equation}
where $|\psi^{SR\leq 2}\rangle$ is an arbitrary quantum pure state with Schmidt rank being no larger than 2.

Take the 2-copy case as an example, recalling operation $E$ defined in Eq.\ref{defe}:
\begin{equation}
    \begin{aligned}
    &E(\rho_w^{T_A}\otimes \rho_w^{T_A})E^{\dagger}= (\frac{1}{d^2 + \beta d})^2 (E(I\otimes I)E^{\dagger} \\
    &+ \beta [E(I\otimes G)E^{\dagger} + E(G\otimes I)E^{\dagger}] + 
    \beta^2 E(G\otimes G)E^{\dagger}).
    \end{aligned}
\end{equation}
The normalization factor of $(\frac{1}{d^2 + \beta d})^2$ can be ignored since it doesn't affect positivity.
The first and the last term remains unchanged in the sense that:
\begin{equation}
    \begin{aligned}
    &E(I\otimes I)E^{\dagger} = E(\sum_{ijkl}|ijkl\rangle
    \langle ijkl|)E^{\dagger} = \sum_{ijkl}|ikjl\rangle\langle ikjl|\\
    &= I(e),
    \end{aligned}
\end{equation}
with $I(e)$ still being the identity matrix, only that it exists on a higher dimensionality of $d^4\times d^4$. Similarly,
\begin{equation}
    \begin{aligned}
    &E(G\otimes G)E^{\dagger} = E(\sum_{ijkl}|iikk\rangle
    \langle jjll|)E^{\dagger} = \sum_{ijkl}|ikik\rangle\langle jljl|\\
    &= G(e),
    \end{aligned}
\end{equation}
with $G(e) = |\Phi(e)\rangle\langle \Phi(e)|$ still being the pure state density matrix of unnormalized 
$|\Phi(e)\rangle = \sum_{ij}|ijij\rangle$, only that it exists on a higher
dimensionality of $d^4\times d^4$.

Regarding the positivity of the middle part $E(I\otimes G)E^{\dagger} + E(G\otimes I)E^{\dagger}$ on $|\psi^{SR\leq 2}\rangle$,
\begin{equation}
    \begin{aligned}
    &E(I\otimes G)E^{\dagger} = E(\sum_{ijkl}|ijkk\rangle
    \langle ijll|)E^{\dagger} = \sum_{ijkl}|ikjk\rangle\langle iljl|\\
    & = d(|\psi_{00}^{s_2=s_4}\rangle\langle \psi_{00}^{s_2=s_4}| + \dots + 
    |\psi_{d-1, d-1}^{s_2=s_4}\rangle\langle \psi_{d-1, d-1}^{s_2=s_4}|)\\
    & = d\sum_{ij}|\psi_{ij}^{s_2=s_4} \rangle\langle \psi_{ij}^{s_2=s_4} |,
    \end{aligned}
\end{equation}
where$|\psi^{s_2=s_4}_{ij}\rangle = \frac{1}{\sqrt{d}}\sum_{k=0}^{d-1}|ikjk\rangle$ is the superposition
of all multipartite states with identical states on the second and fourth subsystems, 
while their first and third subsystems are of states i and j
respectively.
$|\psi^{s_2=s_4}_{ij}\rangle$ is a normalized pure state of Schmidt rank d, 
and all Schmidt coefficients are $\frac{1}{\sqrt{d}}$.
Do a decomposition of $|\psi^{SR\leq 2}\rangle$ in the following way:
\begin{equation}
    |\psi^{SR\leq 2}\rangle = \sum_{ij}p_{ij}|\psi^{SR\leq 2}_{ij}\rangle,
\end{equation}
where $|\psi^{SR\leq 2}_{ij}\rangle$ represents the normalized extraction 
of all $|i-j-\rangle$
terms from the state of 
$|\psi^{SR\leq 2}\rangle$. 
More specifically, for state $|\psi^{SR\leq 2}\rangle$ of:
\begin{equation}
    |\psi^{SR\leq 2}\rangle = \sum_{abcd} p_{abcd} |abcd\rangle,
\end{equation}
the corresponding $|\psi^{SR\leq 2}_{ij}\rangle$ is:
\begin{equation}
    |\psi^{SR\leq 2}_{ij}\rangle = \frac{1}{N_p} \sum_{ab} p_{iajb}|iajb\rangle,
\end{equation}
where $N_p=\sqrt{\sum_{ab}|p_{iajb}|^2}$.

It can be proven that the normalized extraction of all $|ixjy\rangle$ terms must also 
have Schmidt rank $\leq 2$. Consider state-operator isomorphism in Def.\ref{iso},
the Schmidt rank of $|\psi\rangle$ is equivalent to the rank
of the matrix $\Psi(|\psi\rangle)$.
The normalized extraction of all $|i-j-\rangle$ terms composes a new state of $|\psi^{SR\leq 2}_{ij}\rangle$,
therefore, its corresponding matrix is a submatrix of $\Psi(|\psi^{SR\leq 2}\rangle)$,
composed of rows $i\times d+x$ and columns $j\times d+y$, with x,y taking values from $\{0, \dots, d-1\}$.
Since the rank of a submatrix is never larger than the rank of the whole matrix,
$|\psi^{SR\leq 2}_{ij}\rangle$ must have Schmidt rank less than or equal to that of $|\psi^{SR\leq 2}\rangle$.

Using Lem.\ref{fact} again,
it is then clear that:
\begin{equation}
    |\langle\psi_{ij}^{SR\leq 2}|\psi_{pq}^{s_2=s_4}\rangle|^2 
    = |\langle\psi_{ij}^{SR\leq 2}|\psi_{ij}^{s_2=s_4}\rangle|^2\delta_{ip}\delta_{jq}
    \stackrel{L.\ref{fact}}{\leq} \frac{2}{d}\delta_{ip}\delta_{jq}.
\end{equation}
It follows that:
\begin{equation}
    \begin{aligned}
    &\langle\psi^{SR\leq 2}|E(I\otimes G)E^{\dagger}|\psi^{SR\leq 2}\rangle = d\sum_{ij}|p_{ij}|^2 
    |\langle\psi_{ij}^{SR\leq 2}|\psi_{ij}^{s_2=s_4}\rangle|^2 \\
    &\leq 2\sum_{ij}|p_{ij}|^2 = 2.
    \end{aligned}
\end{equation}

A similar analysis is applied to $E(G\otimes I)E^{\dagger}$, getting
\begin{equation}
    \begin{aligned}
    &\langle\psi^{SR\leq 2}|E(G\otimes I)E^{\dagger}|\psi^{SR\leq 2}\rangle = d\sum_{ij}|q_{ij}|^2 
    |\langle\phi_{ij}^{SR\leq 2}|\psi_{ij}^{s_1=s_3}\rangle|^2 \\
    &\leq 2\sum_{ij}|q_{ij}|^2 = 2.
    \end{aligned}
\end{equation}

A sufficient condition for Werner state 2-undistillability is:
\begin{equation}
    \begin{aligned}
    &\langle\psi^{SR\leq 2}|I(e) + \beta^2 G(e) + \beta E(I\otimes G + G\otimes I)E^{\dagger}|\psi^{SR\leq 2}\rangle\\
    &\geq 1 + \beta^2 |\langle\Psi(e)|\psi^{SR\leq 2}\rangle|^2 +4\beta\geq 1+4\beta\geq 0.
    \end{aligned}
\end{equation}
Therefore we can conclude that when $\beta\geq -\frac{1}{4}$, $\rho_w$ must be 2-undistillable.

A similar process is followed in the arbitrary N case, yielding a lower bound regarding all the odd terms in 
$(1+\beta)^{N}$. Since the odd terms in $(1+\beta)^{N}$ can be expressed by $\frac{(1+\beta)^{N} - (1-\beta)^{N}}{2}$,
a lower bound can then be obtained for the N-undistillability case:
\begin{theorem}\label{beta}
    Let $\beta_0$ be the zero point of $1 + (1+\beta)^{N} - (1-\beta)^{N}$ within $[-1, 0]$.
    When $\beta\geq \beta_0$, Werner state $\rho_w$ is N-undistillable.
\end{theorem}
Proof of this is presented in Appendix \ref{beta_proof}.
By calculating the zero point a lower bound is obtained where the corresponding
Werner state is N-undistillable. Note that unlike in previous works, this bound $\beta_0$ is also independent
of dimensionality $d$.
By subsequently raising the dimensionality of Werner states within 
the set, thus raising $-\frac{1}{d}$, we can always obtain a set of Werner states that 
falls into $[\beta_0, -\frac{1}{d})$, which means that they are both 
NPT and N-undistillable for any finite
N. However, as N increases, higher dimensionality are required for undistillability, and so the subset of 
N-undistillable NPT states shrinks to zero as N approaches to infinity, which is
similar to the circumstances encountered in \cite{9910026}, \cite{9910022} and \cite{Chen_2016}. 
Therefore, this method fails to find a state that is undistillable 
for arbitrarily many copies, while being NPT at the same time.

\section{The problem as a partial trace inequality}

The problem of Werner state N-undistillability can be equivalently written as inequalities regarding
the Frobenius norm of partial traces of a rank-2 matrix.

Take 2-undistillability as an example.
$\langle \psi^{SR\leq 2}|E(I\otimes G + G\otimes I)E^{\dagger}|\psi^{SR\leq 2}\rangle$ can be written with regard to 
the Frobenius norms 
of two partial traces:
\begin{equation}
    \langle \psi^{SR\leq 2}|E(I\otimes G)E^{\dagger}|\psi^{SR\leq 2}\rangle = ||Tr_2 (X_2(\psi))||_F^2,
\end{equation}
\begin{equation}
    \langle \psi^{SR\leq 2}|E(G\otimes I)E^{\dagger}|\psi^{SR\leq 2}\rangle = ||Tr_1 (X_2(\psi))||_F^2,
\end{equation}
where $||\cdot||_F$ denotes Frobenius norm:
\begin{equation}
    ||X||_F = \sqrt{\sum_{ij}|X_{ij}|^2} = \sqrt{Tr(X^{\dagger}X)}.
\end{equation}
$X_2(\psi)$ is a $d^2 \times d^2 $ matrix of rank $\leq$ 2, obtained by state-operator
isomorphism, from $|\psi^{SR\leq 2}\rangle$, a pure bipartite $d^2\times d^2$ state with Schmidt rank $\leq$ 2.
For simplicity, we will be writing it as $X_2$ from now on.

The normalization of $|\psi^{SR\leq 2}\rangle$ requires that $||X_2||_F^2 = Tr(X_2^{\dagger}X_2) = 1$.
Now the problem of Werner state 2-undistillability is equivalent to:
Find a range of $\beta$ that makes the
following inequality always holds for arbitrary $X_2$ with rank no larger than 2 (The requirement of 
$||X_2||_F^2=1$ can be lifted due to homogeneity):
\begin{equation}
    \begin{aligned}
    &||X_2||_F^2+\beta^2 |Tr(X_2)|^2 \\
    &+ \beta (||Tr_1(X_2)||_F^2 + ||Tr_2(X_2)||_F^2)\geq 0.
    \end{aligned}
\end{equation}

N-undistillability has a similar form:
\begin{theorem}\label{ineq1}
    Werner state $\rho_{w}$ is N-undistillable iff the following holds for all $d^N\times d^N$ $X_2$:
    \begin{equation}\label{needed}
        \begin{aligned}
            \sum_{S\subset\{1,\cdots,N\}}\beta^{|S|}||Tr_S(X_2)||_F^2\geq 0,
        \end{aligned}
    \end{equation}
    where $Tr_S$ takes partial traces of the subsystems in set $S$.
\end{theorem}
The proof of this is presented in Appendix \ref{ineq1_proof}.
A similar result is obtained by \cite{2310.05726}.

The partial trace inequalities are with regard to a matrix $X_2$ of rank one or two. 
From now on we consider the case where $\beta$ is taken to be $-\frac{1}{2}$ according to popular guess. 
For 2-undistillability problem,
The rank one 
case can be trivially proved by making use of the following lemma:
\begin{lemma}\label{lem1}
The inequalities:
\begin{equation}
    ||Tr_1(X_1)||_F^2 \leq ||X_1||_F^2, ||Tr_2(X_1)||_F^2 \leq ||X_1||_F^2,
\end{equation}
hold for any $d^2\times d^2$ square matrix $X_1$ of rank one.
\end{lemma}
Proof of this lemma is presented in Appendix \ref{lem1_proof}.
For $X_2$ of rank two, the inequality of concern when $\beta = -\frac{1}{2}$
is:
\begin{equation}\label{reformed}
    ||Tr_2(X_2)||_F^2 + ||Tr_1(X_2)||_F^2 - \frac{1}{2}|Tr(X_2)|^2 \leq 2||X_2||_F^2 = 2.
\end{equation}
For any square matrix $X_2$ of rank two, the singular decomposition can be applied, decomposing the matrix as 
the sum of two rank one matrices:
\begin{equation}
    X_2 = \sigma_1 X_2^1 + \sigma_2 X_2^2 = \sigma_1 \mathbf{u}^1\mathbf{v}^{1\dagger} + \sigma_2 \mathbf{u}^2\mathbf{v}^{2\dagger},
\end{equation}
where $\sigma_1^2 + \sigma_2^2 = 1$, and both $\sigma_1$ and $\sigma_2$ are positive.
We introduce $d\times d$ matrices $U_1, U_2, V_1, V_2$ that are the result of state-operator isomorphisms,
corresponding to $\mathbf{u}^1, \mathbf{u}^2, \mathbf{v}^1, \mathbf{v}^2$ respectively.

Eq.\ref{reformed} then becomes:
\begin{equation}
    \begin{aligned}
    &\sigma_1^2 (Tr(V_1 U_1^{\dagger} U_1 V_1^{\dagger}) + Tr(U_1^{\dagger} V_1 V_1^{\dagger} U_1))\\
    &+\sigma_2^2 (Tr(V_2 U_2^{\dagger} U_2 V_2^{\dagger}) + Tr(U_2^{\dagger} V_2 V_2^{\dagger} U_2))\\
    &+\sigma_1 \sigma_2 2 Re[Tr(V_1 U_1^{\dagger} U_2 V_2^{\dagger}) + Tr(U_1^{\dagger} V_1 V_2^{\dagger} U_2)]\\
    &-\frac{|\sigma_1 Tr(U_1 V_1^{\dagger}) + \sigma_2 Tr(U_2 V_2^{\dagger})|^2}{2} \leq 2.
    \end{aligned}
\end{equation}
For simplicity the above is rewritten as:
\begin{equation}
    \sigma_1^2 P + \sigma_2^2 Q + \sigma_1\sigma_2 R \leq 2,
\end{equation}
where $P, Q, R$ are:
\begin{equation}
    P = Tr(U_1^{\dagger}V_1V_1^{\dagger}U_1) + Tr(V_1U_1^{\dagger}U_1V_1^{\dagger})-\frac{|Tr(U_1V_1^{\dagger}) |^2}{2},
\end{equation}
\begin{equation}
    Q = Tr(U_2^{\dagger}V_2V_2^{\dagger}U_2) + Tr(V_2U_2^{\dagger}U_2V_2^{\dagger})-\frac{|Tr(U_2V_2^{\dagger}) |^2}{2},
\end{equation}
\begin{equation}
    \begin{aligned}
    &R = 2Re[Tr(V_1U_1^{\dagger}U_2V_2^{\dagger}) + Tr(U_1^{\dagger}V_1V_2^{\dagger}U_2)\\ 
    &- \frac{Tr^*(U_1V_1^{\dagger})Tr(U_2V_2^{\dagger})}{2}].
    \end{aligned}
\end{equation}
Maximization regarding variables $\sigma_1,\sigma_2$ reduces the question to proving:
\begin{equation}\label{proved}
    R^2 \leq 4(2-P)(2-Q),
\end{equation}
with normalization and orthogonality conditions requiring:
\begin{equation}
    ||U_1||_F^2 = ||V_1||_F^2 = ||U_2||_F^2 = ||V_2||_F^2=1,
\end{equation}
and
\begin{equation}
    tr(V_2^{\dagger}V_1) = tr(U_2^{\dagger}U_1) = 0.
\end{equation}

The special case of $U_1 = V_1$, in other words, $X_2$ being the sum of a normal matrix and a rank-1 matrix,
has been proved in \cite{2310.05726}.

In fact, by a slight change of representation, we can also get an equivalent expression of
Werner state N-undistillability inequality, regarding only rank one matrices:

\begin{theorem}\label{ineq2}
    Werner state $\rho_{w}$ is N-undistillable iff
    \begin{equation}
        Re[f_N(X_1, X'_1)]^2 \leq f_N(X'_1, X'_1)f_N(X_1, X_1),
      \end{equation}
      where
      \begin{equation}
        \begin{aligned}
          &f_N(X_1, X'_1) = \sum_{S\subset\{1,\cdots,N\}}\beta^{|S|} Tr[Tr^{\dagger}_{S}(X_1) Tr_{S}(X'_1)]\geq 0,
        \end{aligned}
      \end{equation}
      where $X_1=\mathbf{w}^{\dagger}\mathbf{x}, X'_1=\mathbf{y}^{\dagger}\mathbf{z}$
      are $d^N\times d^N$ rank one matrices, and their component vectors $\mathbf{w},\mathbf{x},\mathbf{y},
      \mathbf{z}$ satisfy $\mathbf{w}\perp\mathbf{y},\mathbf{x}\perp\mathbf{z}$.
\end{theorem}
Proof of this equivalence is presented in Appendix \ref{ineq2_proof}.

It is then established that N-undistillability problem for arbitrary N can be written as inequalities
concerning rank one matrices.
Although the above form looks like a Cauchy-Schwartz inequality, 
$f_N(X_1, X'_1)$ cannot be seen as an inner product and enable the direct application of Cauchy-Schwartz inequality.
$f_N(X_1, X'_1)$ is obviously conjugate symmetric, and 
for rank one matrices $X_1$ and $X'_1$, positivity of $f_N(X_1, X_1)$ and $f_N(X'_1, X'_1)$ 
can be ascertained. However, rank-one matrices
do not compose a vector space, since the sum of two rank-one matrices can be a rank-two matrix. Although arbitrary 
finite rank matrices do compose a finite dimensional vector space, proving positivity of $f_N(A, A)$ for arbitrary rank matrix $A$ is both
beyond our ability and our need, since proving positivity of $f_N(X_2, X_2)$ for arbitrary rank two matrix $X_2$ is, in fact, 
equivalent to Eq.\ref{needed}, and is enough for proof of N-undistillability.

\section{Conversion to a multi-variable function}

We make another attempt at this problem by seeing it as a multi-variable function.
For simplicity we consider only the real case of the 2-distillability problem,
that is, $X_1$, $X'_1$ are $d^2\times d^2$ real matrices.
For convenience of expression we use $f(C, D) = f_2(X_1, X'_1)$ to denote the function concerned.
Thus, the 2-undistillability problem 
is equivalent to proving the following inequality for all $d^2\times d^2$-dimensional rank one matrices $C, D$:
\begin{equation}
    f^2(C, D) \leq f(C, C) f(D, D),
\end{equation}
where 
\begin{equation}
    \begin{aligned}
    &f(C, D) = Tr(C^{T}D) \\
    &+\beta (Tr(C_1^{T}D_1) + Tr(C_2^{T}D_2)) + \beta^2 Tr(C^{T})Tr(D),
    \end{aligned}
\end{equation}
where
we have used $Tr_1(C)=C_2, Tr_2(C)=C_1, Tr_1(D)=D_2, Tr_2(D)=D_1$ for simplicity.
A multi-variable function $g(C)$ is defined if we see $D$ as a constant $D_0$:
\begin{equation}
    g(C)_{D_0} = f(C, C)f(D_0, D_0) - f^2(C, D_0),
\end{equation}
where $D_0$ is an arbitrary $d^2\times d^2$-dimensional real matrix of rank 1. 
Both $C$ and $D_0$ can be written as outer products of $d^2$-dimensional real vectors 
$\mathbf{w}$, $\mathbf{x}, \mathbf{y}, \mathbf{z}$:
\begin{equation}
    C = \mathbf{w}\mathbf{x}^T, D_0 = \mathbf{y}\mathbf{z}^T.
\end{equation}
The multi variables of function $g(C)$ are taken to be the 
components of vector $\mathbf{w}, \mathbf{x}$, represented 
by $w_{ij}, x_{ij}$. 
Indices i and j take their values in $[0, 1,\cdots, d-1]$, 
and are combined together to represent the $(i\times d + j)$-th component of the vector.
The problem is therefore transformed into proving that 
multi-variable function $g(C)_{D_0}$ is always non-negative, 
for all variables $w_{ij}, x_{ij}$ 
and all possible parameters $y_{ij}, z_{ij}$.
It is clear that when 
$C = D_0$, $g(D_0)_{D_0} = f^2(D_0, D_0) - f^2(D_0, D_0) = 0$, 
and we conjecture that these are global minimums.

We first prove that the gradients 
at these points are zero 
regardless of the parameters $\mathbf{y}, \mathbf{z}$ by calculating the Jacobian matrix, therefore
showing that the $C=D_0$ points are indeed critical points.
Then the Hessian matrix at this point is presented, in particular its three block parts, since
the Hessian matrix is Hermitian. We conjecture that it is positive semidefinite.
An additional proof of nonconvexity is also given in Appendix.\ref{non_con_proof}, thus eliminating the easy case where
any local minimum is also a global minimum.

\subsection{The gradient at \texorpdfstring{$C = D_0$}{Lg} is zero}

For simplicity all $g(C)_{D_0}$ are written as $g(C)$ in the following text.
Define 
\begin{equation}
    \begin{aligned}
    &h_{1, ij}(w,x) = 2w_{ij}\left(\sum_{pq} x_{pq}^2\right)\\
    &+ 2\beta\left[\sum_{pq}x_{pj}w_{iq}x_{pq} + \sum_{pq}x_{iq}w_{pj}x_{pq}\right] \\
    &+ 2 \beta^2 x_{ij} \left(\sum_{pq} w_{pq}x_{pq}\right)
    \end{aligned}
\end{equation}
\begin{equation}
    \begin{aligned}
    &h_{2, ij}(y, z, x) = y_{ij}\left(\sum_{pq} x_{pq}z_{pq}\right)\\
    &+ \beta\left[\sum_{pq}x_{pj}y_{iq}z_{pq} + \sum_{pq}x_{iq}y_{pj}z_{pq}\right] \\
    &+  \beta^2 x_{ij}\left(\sum_{pq} y_{pq}z_{pq}\right)
\end{aligned}
\end{equation}
The partial derivatives of $f(C, C)$ and $f(C, D_0)$ are:
\begin{equation}
    \frac{\partial f(C, C)}{\partial w_{ij}}=h_{1, ij}(w,x), \frac{\partial f(C, C)}{\partial x_{ij}}
    =h_{1, ij}(x,w)
\end{equation}
\begin{equation}
    \frac{\partial f(C, D_0)}{\partial w_{ij}}= h_{2, ij}(y, z, x), \frac{\partial f(C, D_0)}{\partial x_{ij}}
    =h_{2, ij}(z, y, w)
\end{equation}

Notice that $h_{1, ij}(w, x) = 2 h_{2, ij}(w, x, x), h_{1, ij}(x, w) = 2 h_{2, ij}(x, w, w)$.

Therefore, at the critical points where $C = D_0$,
\begin{equation}
    \frac{\partial f(C, D_0)}{\partial w_{ij}}|_{C = D_0} = \frac{1}{2}
    \frac{\partial f(C, C)}{\partial w_{ij}}|_{C = D_0},
\end{equation}
\begin{equation}
    \frac{\partial f(C, D_0)}{\partial x_{ij}}|_{C = D_0} = \frac{1}{2}
    \frac{\partial f(C, C)}{\partial x_{ij}}|_{C = D_0}.
\end{equation}
It then follows that the Jacobian at the $C = D_0$ point is zero:
\begin{equation}
    \begin{aligned}
    &\frac{\partial g(C)}{\partial w_{ij}}|_{C = D_0} = f(D_0, D_0) \frac{\partial f(C, C)}
    {\partial w_{ij}}|_{C = D_0}
    \\
    &-2\frac{\partial f(C, D_0)}{\partial w_{ij}}|_{C = D_0} f(D_0, D_0) = 0.
\end{aligned}
\end{equation}
\begin{equation}
    \begin{aligned}
    &\frac{\partial g(C)}{\partial x_{ij}}|_{C = D_0} = f(D_0, D_0) \frac{\partial f(C, C)}
    {\partial x_{ij}}|_{C = D_0}
    \\
    &-2\frac{\partial f(C, D_0)}{\partial x_{ij}}|_{C = D_0} f(D_0, D_0) = 0.
\end{aligned}
\end{equation}
It is then established that the gradients at the $C = D_0$ points are zero, and the points are
indeed critical points.

\subsection{Hessian matrix}

Hessian matrices at the critical points are presented.
Define:
\begin{equation}
    \begin{aligned}
    &h_{3, ijkl}(x) = 2\delta_{ik}\delta_{jl}\left(\sum_{pq} x^2_{pq}\right)\\
    & + 2\beta\left[\delta_{ik}\sum_{p}x_{pj}x_{pl} + \delta_{lj}\sum_q x_{iq} x_{kq}\right] \\
    &+ 2 \beta^2 x_{ij}x_{kl}
\end{aligned}
\end{equation}
\begin{equation}
    \begin{aligned}
    &h_{4, ijkl}(w, x) = 4w_{ij}x_{kl} \\
    &+ \beta\left[2\delta_{jl}(\sum_q w_{iq}x_{kq}) + x_{kj}w_{il} + 2\delta_{ik}(\sum_p w_{pj} x_{pl}) + x_{il}w_{kj}\right]\\
    & + 2\beta^2\left[\delta_{ik}\delta_{jl}\left(\sum_{pq}w_{pq}x_{pq}\right) + x_{ij}w_{kl}\right]
\end{aligned}
\end{equation}
\begin{equation}
    \begin{aligned}
    &h_{5, ijkl}(y, z) = y_{ij}z_{kl} \\
    &+ \beta\left[\delta_{jl}(\sum_q y_{iq}z_{kq}) +  \delta_{ik}(\sum_p y_{pj} z_{pl})\right]\\
    & + \beta^2 \delta_{ik}\delta_{jl}\left(\sum_{pq}y_{pq}z_{pq}\right)
\end{aligned}
\end{equation}
The Hessian matrices at $C=D_0$ can be explicitly written as:
\begin{equation}
    \begin{pmatrix}
        (\frac{\partial^2 g(C)}{\partial w_{ij}\partial w_{kl}})|_{C=D_0} & (\frac{\partial^2 g(C)}{\partial w_{ij}\partial x_{kl}})|_{C=D_0}\\
        (\frac{\partial^2 g(C)}{\partial x_{ij}\partial w_{kl}})|_{C=D_0} & (\frac{\partial^2 g(C)}{\partial x_{ij}\partial x_{kl}})|_{C=D_0} \\
    \end{pmatrix}.
\end{equation}
where the three independent parts of Hessian matrix are: 
\begin{equation}
    \begin{aligned}
    &\frac{\partial^2 g(C)}{\partial w_{ij}\partial w_{kl}}|_{C=D_0}  \\
    &= f(D_0, D_0)h_{3, ijkl}(x)
    -2h_{2, ij}(w, x, x)h_{2, kl}(w, x, x),
\end{aligned}
\end{equation}
\begin{equation}
    \begin{aligned}
    &\frac{\partial^2 g(C)}{\partial x_{ij}\partial x_{kl}}|_{C=D_0} \\
    &= f(D_0, D_0)h_{3, ijkl}(w)
    -2h_{2, ij}(x, w, w)h_{2, kl}(x, w, w),
\end{aligned}
\end{equation}

\begin{equation}
    \begin{aligned}
    &\frac{\partial^2 g(C)}{\partial w_{ij}\partial x_{kl}}|_{C=D_0} = f(D_0, D_0)h_{4, ijkl}(w, x)\\
    &-2h_{5, ijkl}(w, x)f(D_0, D_0) - 2h_{2, ij}(w, x, w)h_{2, kl}(x, w, w).
\end{aligned}
\end{equation}

We conjecture that this Hessian matrix is positive semidefinite, thus making the critical points
local minimums. In Appendix \ref{non_con_proof} we prove that the function is non-convex,
which means that further scrutinization is needed for characterization of this function.

\section*{Conclusion and Discussion}
We've broken down the process of verifying bound entanglement into iterative steps of 
1-undistillability verifications, by noticing that N-undistillability for arbitrary N is
equivalent to $2^N$-undistillability for arbitrary N, and then utilizing specific symmetric properties.
In the Werner state case, 
new bounds for N-undistillability of any finite N are presented,
a result similar to that of \cite{9910026}, \cite{9910022} and \cite{Chen_2016}, 
but different in the sense that our bounds are unaffected 
by the dimensionality of Werner states. Alternative 
expressions for inequalities applicable to both rank two and rank one matrices are given. 
Subsequently, the problem of two-undistillability is converted into a matrix 
analysis problem. 
The multi-variable 
function treatment is also attempted, as well as proving critical points, non-convexity and conjecturing
about Hessian positivity. 

Recently, the problem of N-copy Werner state distillability has been reformulated in \cite{2310.05726} into
a set of partial trace inequalities, which coincided with our Theorem \ref{ineq1}. 
In the regime of 2-copy Werner state distillability,
A special case of the matrix being the sum of a rank-one matrix
and a normal matrix is proved. 

We believe the new perspectives presented here will
assist further attempts at this famous open problem.

\acknowledgments
We would like to thank Prof. Fedor Sukochev, Dr. Dmitriy Zanin and Prof. Zhi Yin for helpful comments and valuable insights about partial trace inequalities.

\bibliography{NPT_bound}

\appendix

\section{Proof of Lemma \ref{fact}}\label{fact_proof}

\begin{lemma}
    Let $|\psi\rangle$ be an arbitrary pure quantum state and
        $s_j$ are the Schmidt coefficients of $|\psi\rangle$ arranged in descending order, then 
        \begin{equation}
           \;\sum_{j=1}^k s_j^2= max_{|\phi^{SR\leq k}\rangle} |\langle \psi|\phi^{SR\leq k}\rangle |^2,
        \end{equation}
        where $|\phi^{SR\leq k}\rangle$ is an arbitrary quantum pure state with Schimidt rank being less than k, namely, $SR(|\phi^{SR\leq k}\rangle)\leq k$.
\end{lemma}

\begin{proof}
State-operator isomorphism translates the problem to an equivalent form of proving:
\begin{equation}
    |Tr(A^{\dagger}B)|^2 \leq \sum_{i=1}^k s^2_i(A),
\end{equation}
where both $A$ and $B$ are $d\times d$ matrices, and as they are the result of state-operator isomorphism
from $|\psi\rangle$ and $|\phi^{SR\leq 2}\rangle$,
are of Frobenius norm 1 due to normalization conditions, with $B$ having rank k. 
We use $s_i(A)$ and $|\lambda_i(A)|$ 
to denote the singular values and absolute values of eigenvalues
of $A$ arranged in descending order.
\begin{equation}
    \begin{aligned}
    &|Tr(A^{\dagger}B)|^2 \leq (\sum_{i=1}^d |\lambda_i(A^{\dagger}B)|)^2 \leq (\sum_{i=1}^d s_i(A^{\dagger}B))^2\\
    &\leq (\sum_{i=1}^k s_i(A^{\dagger})s_i(B))^2 = (\sum_{i=1}^k s_i(A)s_i(B))^2,
    \end{aligned}
\end{equation}
where the second and third inequality are from Theorem 3.3.13(a) and Theorem 3.3.14(a) in \cite{Horn_Johnson_1991},
and $d\to k$ change of index holds because the rest of $B$'s singular values are all zeroes after $d$. 
respectively.
Then by Cauchy-Schwartz inequality,
\begin{equation}
    \begin{aligned}
    &|Tr(A^{\dagger}B)|^2 \leq (\sum_{i=1}^k s_i(A)s_i(B))^2 \leq (\sum_{i=1}^k s^2_i(A))(\sum_{j=1}^k s^2_j(B))\\
    &=\sum_{i=1}^k s^2_i(A),
    \end{aligned}
\end{equation}
where the last equality follows from the normalization condition of $B$.
\end{proof}

\section{Proof of Theorem \ref{beta}}\label{beta_proof}
\begin{theorem}
    Let $\beta_0$ be the zero point of $1 + (1+\beta)^{N} - (1-\beta)^{N}$ within $[-1, 0]$.
    When $\beta\geq \beta_0$, Werner state $\rho_w$ is N-undistillable.
\end{theorem}
\begin{proof}
We can overlook the $\frac{1}{d^2+\beta d}$ factor since it doesn't affect positivity.
\begin{equation}
    (\rho_w^{T_A})^{\otimes N} = I^{\otimes N} + \sum_{m=1}^N \beta^{m} \sum_{seq(m,N)}
    G^{i_1}\otimes G^{i_2}\otimes \cdots \otimes G^{i_N},
\end{equation}
where $seq(m,N)$ denotes all possible binary $i_1, \dots, i_N$ sequences with $m$ ones and $N-m$ zeros. 
For convenience denote $S_0 = \{n|i_n=0\},S_1 = \{n|i_n=1\}$.
\begin{footnotesize}
\begin{equation}
    \begin{aligned}
    &M_N G^{i_1}\otimes G^{i_2}\otimes \cdots \otimes G^{i_N} M_N^{\dagger}\\
    &=M_N \sum_{j_1 \dots j_{2N}} |\cdots j_{2n-1}j_{2n-i_n} \cdots\rangle \langle
    \cdots j_{2n-1+i_n} j_{2n} \cdots | M_N^{\dagger}\\
    &=\sum_{j_1 \dots j_{2N}} |\cdots j_{2n-1} \cdots,\cdots j_{2n-i_n} \cdots \rangle\langle
    \cdots j_{2n-1+i_n} \cdots,\cdots j_{2n} \cdots|\\
    &=d^m \sum_{j_{2n-1},j_{2n}, n\in S_0} |\psi_{j_{2n-1},j_{2n}, n\in S_0}\rangle
    \langle \psi_{j_{2n-1},j_{2n}, n\in S_0}|,
    \end{aligned}
\end{equation}
\end{footnotesize}
where
\begin{equation}
    \begin{aligned}
    &|\psi_{j_{2n-1},j_{2n}, n\in S_0}\rangle \\
    &= \frac{1}{\sqrt{d^m}} \sum_{j_{2n-1}, n\in S_1}
    |\dots j_{2n-1} \dots, \dots j_{2n-i_n} \dots \rangle.
    \end{aligned}
\end{equation}
Any pure quantum state $|\psi^{SR\leq 2}\rangle$ of Schmidt rank no larger than 2 can be decomposed into:
\begin{equation}
    |\psi^{SR\leq 2}\rangle = \sum_{j_{2n-1},j_{2n}, n\in S_0} p_{j_{2n-1},j_{2n}, n\in S_0}
    |\psi_{j_{2n-1},j_{2n}, n\in S_0}^{SR\leq 2}\rangle,
\end{equation}
where $|\psi_{j_{2n-1},j_{2n}, n\in S_0}^{SR\leq 2}\rangle$ is the normalized extraction of 
all $|\cdots j_{2n-1}\cdots, \cdots j_{2n}\cdots \rangle, n\in S_0$ terms in $|\psi^{SR\leq 2}\rangle$,
namely, if 
\begin{equation}
    |\psi^{SR\leq 2}\rangle = \sum_{k_1, \cdots, k_{2N}}c_{k_1, \cdots, k_{2N}}|k_1, \cdots, k_{2N}\rangle,
\end{equation}
then,
\begin{equation}
    \begin{aligned}
    &|\psi^{SR\leq 2}_{j_{2n-1},j_{2n},n\in S_0}\rangle \\
    &= \frac{1}{p_{j_{2n-1},j_{2n},n\in S_0}}\\
    &\sum_{j_{2n-1},j_{2n},n\in S_1} c_{\cdots j_{2n-1}\cdots, \cdots j_{2n}\cdots}
    |\cdots j_{2n-1}\cdots, \cdots j_{2n}\cdots\rangle,
    \end{aligned}
\end{equation}
where
\begin{equation}
    p_{j_{2n-1},j_{2n},n\in S_0} = \sqrt{\sum_{j_{2n-1},j_{2n},n\in S_1} |c_{\cdots j_{2n-1}\cdots, \cdots j_{2n}\cdots} |^2}
\end{equation}
It can be proved that $|\psi_{j_{2n-1},j_{2n}, n\in S_0}^{SR\leq 2}\rangle$ also has Schmidt rank
less than or equal to 2. According to state-operator isomorphism, Schmidt rank of a state is equal
to rank of the corresponding operator. In fact, the operator corresponding to 
$|\psi_{j_{2n-1},j_{2n}, n\in S_0}^{SR\leq 2}\rangle$ is the extraction of $d^m$
rows $\cdots j_{2n-1}\cdots, n\in S_0,$ and $d^m$ columns $\cdots j_{2n}\cdots, n\in S_0$. 
The row and column numbers
$\cdots j_{2n-1}\cdots$ and $\cdots j_{2n}\cdots$ are written in base-d numeral systems.
Therefore, operator corresponding to $|\psi_{j_{2n-1},j_{2n}, n\in S_0}^{SR\leq 2}\rangle$
is a submatrix of the operator corresponding to $|\psi^{SR\leq 2}\rangle$.
The rank of a submatrix is no larger than the rank of the entire matrix, therefore 
$|\psi_{j_{2n-1},j_{2n}, n\in S_0}^{SR\leq 2}\rangle$ must have Schmidt rank no larger than 2.
Using Lem.\ref{fact}, 
it is then clear that:
\begin{equation}
    |\langle\psi_{j_{2n-1},j_{2n}, n\in S_0}^{SR\leq 2}|\psi_{j_{2n-1},j_{2n}, n\in S_0}\rangle|^2
    \leq \frac{2}{d^m},
\end{equation}
and so 
\begin{equation}
    \begin{aligned}
    &\langle \psi^{SR\leq 2} | M_N G^{i_1}\otimes G^{i_2}\otimes \cdots \otimes G^{i_N} M_N^{\dagger} |
    \psi^{SR\leq 2} \rangle\\
    &\leq d^m \frac{2}{d^m} \sum_{j_{2n-1},j_{2n}, n\in S_0} |p_{j_{2n-1},j_{2n}, n\in S_0}|^2 = 2.
    \end{aligned}
\end{equation}
Considering the fact that all even terms in $(\rho_w^{T_A})^{\otimes N}$ are non-negative
and using the above inequality on all odd terms, a lower bound is obtained:
\begin{equation}
    \begin{aligned}
    &\langle \psi^{SR\leq 2} |M_N(\rho_w^{T_A})^{\otimes N}M_N^{\dagger}|\psi^{SR\leq 2} \rangle\\
    &\geq 1+\sum_{k=1}^{\lceil \frac{N}{2}\rceil}2 C_N^{2k-1} \beta^{2k-1} = 1 + (1+\beta)^N - (1-\beta)^N
    \end{aligned}
\end{equation}
For $\beta\in (-1, 0)$, any $\beta$ greater than the zero point $\beta_0$ would make
$1 + (1+\beta)^N - (1-\beta)^N$ greater than zero, thus ensuring the positivity of 
$\langle \psi^{SR\leq 2} |M_N(\rho_w^{T_A})^{\otimes N}M_N^{\dagger}|\psi^{SR\leq 2} \rangle$.

\end{proof}

\section{Proof of Theorem \ref{ineq1}}\label{ineq1_proof}
\begin{theorem}
    Werner state $\rho_{w}$ is N-undistillable iff the following holds for all $d^N\times d^N$ $X_2$:
    \begin{equation}
        \begin{aligned}
            \sum_{S\subset\{1,\cdots,N\}}\beta^{|S|}||Tr_S(X_2)||_F^2\geq 0,
        \end{aligned}
    \end{equation}
    where $Tr_S$ takes partial traces of the subsystems in set $S$.
\end{theorem}

\begin{proof}
    We've already established that N-undistillability is equivalent to:
    \begin{equation}
        \langle \psi^{SR\leq 2} | M_N (\rho^{T_A})^{\otimes N} M_N^{\dagger} |\psi^{SR\leq2}\rangle \geq 0.
    \end{equation}
In the previous section it has been established that:
\begin{equation}
    (\rho_w^{T_A})^{\otimes N} = I^{\otimes N} + \sum_{m=1}^N \beta^{m} \sum_{seq(m,N)}
    G^{i_1}\otimes G^{i_2}\otimes \cdots \otimes G^{i_N},
\end{equation}
\begin{equation}
    \begin{aligned}
    &M_N G^{i_1}\otimes G^{i_2}\otimes \cdots \otimes G^{i_N} M_N^{\dagger}\\
    &=d^m \sum_{j_{2n-1},j_{2n}, n\in S_0} |\psi_{j_{2n-1},j_{2n}, n\in S_0}\rangle
    \langle \psi_{j_{2n-1},j_{2n}, n\in S_0}|,
    \end{aligned}
\end{equation}
where
\begin{equation}
    \begin{aligned}
    &|\psi_{j_{2n-1},j_{2n}, n\in S_0}\rangle \\
    &= \frac{1}{\sqrt{d^m}} \sum_{j_{2n-1}, n\in S_1}
    |\dots j_{2n-1} \dots, \dots j_{2n-i_n} \dots \rangle.
    \end{aligned}
\end{equation}
For a pure quantum state of the form:
\begin{equation}
    |\psi^{SR\leq 2}\rangle = \sum_{j_1, \cdots, j_N}
    c_{\cdots j_{2n-1}\cdots, \cdots j_{2n}\cdots}|\cdots j_{2n-1}\cdots, \cdots j_{2n}\cdots\rangle,
\end{equation}
\begin{equation}
    \begin{aligned}
        &\langle \psi^{SR\leq 2} |M_N G^{i_1}\otimes G^{i_2}\otimes \cdots \otimes G^{i_N} M_N^{\dagger}
        |\psi^{SR\leq2}\rangle\\
        &=d^m \sum_{j_{2n-1}, j_{2n}, n\in S_0} |\langle\psi_{j_{2n-1}, j_{2n}, n\in S_0}|\psi^{SR\leq 2}\rangle|^2\\
        &=\sum_{j_{2n-1}, j_{2n}, n\in S_0} |\sum_{j_{2n-1}, n\in S_1} c_{\cdots j_{2n-1} \cdots, \cdots j_{2n-i_n}
        \cdots}|^2\\
        &=||Tr_{\{i_n | n\in S_1\}}(X_2)||_F^2.
    \end{aligned}
\end{equation}
In the last equality, $X_2$ is the result of state-operator isomorphism from $|\psi^{SR\leq 2}\rangle$.
Noticing that
\begin{equation}
    \langle \psi^{SR\leq 2} |M_N I^{\otimes N} M_N^{\dagger}|\psi^{SR\leq 2}\rangle = 1 = \beta^0 ||X_2||_F^2,
\end{equation}
it then follows that
\begin{equation}
    \langle \psi^{SR\leq 2} | M_N (\rho^{T_A})^{\otimes N} M_N^{\dagger} |\psi^{SR\leq2}\rangle \geq 0,
\end{equation}
is equivalent to
\begin{equation}
    \begin{aligned}
        \sum_{S\subset\{1,\cdots,N\}}\beta^{|S|}||Tr_S(X_2)||_F^2 \geq 0,
    \end{aligned}
\end{equation}
\end{proof}

\section{Proof of Theorem \ref{ineq2}}\label{ineq2_proof}

\begin{theorem}
    Werner state $\rho_{w}$ is N-undistillable iff
    \begin{equation}
        Re[f_N(X_1, X'_1)]^2 \leq f_N(X'_1, X'_1)f_N(X_1, X_1),
      \end{equation}
      where
      \begin{equation}
        \begin{aligned}
          &f_N(X_1, X'_1) = \sum_{S\subset\{1,\cdots,N\}}\beta^{|S|} Tr[Tr^{\dagger}_{S}(X_1) Tr_{S}(X'_1)]\geq 0,
        \end{aligned}
      \end{equation}
      where $X_1=\mathbf{w}^{\dagger}\mathbf{x}, X'_1=\mathbf{y}^{\dagger}\mathbf{z}$
      are $d^N\times d^N$ rank one matrices, and their component vectors $\mathbf{w},\mathbf{x},\mathbf{y},
      \mathbf{z}$ satisfy $\mathbf{w}\perp\mathbf{y},\mathbf{x}\perp\mathbf{z}$.
\end{theorem}
\begin{proof}
    It has been proven that N-undistillability is equivalent to:
    \begin{equation}\label{pre}
        \begin{aligned}
            \sum_{S\subset\{1,\cdots,N\}}\beta^{|S|}||Tr_S(X_2)||_F^2\geq 0,
        \end{aligned}
    \end{equation}
    Since any rank two matrix can be decomposed via singular value decomposition:
    \begin{equation}
        X_2 = \sigma_1 X_1 +\sigma_2 X'_1,
    \end{equation}
    Eq.\ref{pre} can be written as:

    \begin{equation}\label{c5}
        \begin{aligned}
            &\sum_{S\subset\{1,\cdots,N\}}\beta^{|S|}||Tr_S(X_2)||_F^2\\
            &=\sum_{S\subset\{1,\cdots,N\}}\beta^{|S|}[\sigma^2_1 ||Tr_{S}(X_1)||_F^2+\sigma_2^2||Tr_{S}(X'_1)||_F^2\\
            &+2\sigma_1\sigma_2 Re(Tr(Tr_S(X_1)^{\dagger}Tr_S(X'_1)))]\\
            &=\sigma_1^2 f_N(X_1, X_1) + \sigma_2^2 f_N(X'_1, X'_1) \\
            &+ 2\sigma_1\sigma_2 Re[f_N(X_1, X'_1)]\\
            &\geq 0. \\
        \end{aligned}
    \end{equation}
    The above should hold for all singular values $\sigma_1,\sigma_2$, which is then
    equivalent to:
    \begin{equation}
        Re[f_N(X_1, X'_1)]^2 \leq f_N(X_1, X_1)f_N(X'_1, X'_1),
    \end{equation}
    always holding when $X_1=\mathbf{w}^{\dagger}\mathbf{x}, X'_1=\mathbf{y}^{\dagger}\mathbf{z}$
    are $d^N\times d^N$ rank one matrices, and their component vectors $\mathbf{w},\mathbf{x},\mathbf{y},
    \mathbf{z}$ satisfy $\mathbf{w}\perp\mathbf{y},\mathbf{x}\perp\mathbf{z}$.
\end{proof}

\section{Proof of Lem. \ref{lem1}}\label{lem1_proof}

\begin{lemma}
    The inequalities:
    \begin{equation}
        ||Tr_1(X_1)||_F^2 \leq ||X_1||_F^2, ||Tr_2(X_1)||_F^2 \leq ||X_1||_F^2,
    \end{equation}
    hold for any $d^2\times d^2$ square matrix $X_1$ of rank one.
    \end{lemma}
\begin{proof}
A rank one matrix $X_1$ can always be written as the outer product of two vectors: (in this case, bipartite)
\begin{equation}
    (X_1)_{ij,kl} = w_{ij}x_{kl}.
\end{equation}
Its partial traces can be computed accordingly:
\begin{equation}
    (Tr_2(X_1))_{i,k} = \sum_j w_{ij}x_{kj}, (Tr_1(X_1))_{j,l} = \sum_i w_{ij}x_{il}.
\end{equation}
A direct calculation and application of Cauchy-Schwartz inequality yields the desired result:
\begin{equation}
    \begin{aligned}
    &||Tr_2(X_1)||_F^2 = \sum_{ik}|\sum_j w_{ij}x_{kj}|^2 \leq \sum_{ik} \sum_j |w_{ij}|^2\sum_l |x_{kl}|^2\\
    &=\sum_{ijkl}|w_{ij}|^2|x_{kl}|^2 = ||X_1||_F^2,
    \end{aligned}
\end{equation}
\begin{equation}
    \begin{aligned}
    &||Tr_1(X_1)||_F^2 = \sum_{jl}|\sum_i w_{ij}x_{il}|^2 \leq \sum_{jl} \sum_i |w_{ij}|^2\sum_k |x_{kl}|^2\\
    &=\sum_{ijkl}|w_{ij}|^2|x_{kl}|^2 = ||X_1||_F^2.
    \end{aligned}
\end{equation}
\end{proof}

\section{Proving non-convexity}\label{non_con_proof}

We now prove that the function of $g(C)_{D_0}$ is non-convex,
by showing that its local minimum set is not a convex one.
For 
\begin{equation}
    C = \mathbf{w}\mathbf{x}^T, D_0 = \mathbf{y}\mathbf{z}^T,
\end{equation}
we set vectors $\mathbf{y}$ and $\mathbf{z}$ to identical forms of:
\begin{equation}
    y_{ij} = z_{ij} = \delta_{i0}\delta_{j1}
\end{equation}
For simplicity we first write the variables $\mathbf{w}, \mathbf{x}$ in a matrix form of:
$$
\begin{pmatrix}
 w_{00} & w_{01} & \cdots & w_{d-1,d-1}\\
 x_{00} & x_{01} & \cdots & x_{d-1,d-1}\\
\end{pmatrix},
$$
where the two rows correspond to two vectors $\mathbf{w}$, $\mathbf{x}$, respectively.
We take the middle point combination of the following two points:
$$
\begin{pmatrix}
 0 & 1 & \cdots 0\\
 0 & 1 & \cdots 0\\
\end{pmatrix},
\begin{pmatrix}
  1 & 0 & \cdots 0\\
  1 & 0 & \cdots 0\\
 \end{pmatrix},
$$
both of these points can be proven to have zero Jacobian matrix and positive definite Hessian, effectively
making them local minimums.

For a convex function, its local minimum set must be a convex set, so that middle point combination
of any two points should stay in the set.
Their middle point combination is:
$$
\begin{pmatrix}
 1 & 1 & 0 &  \cdots 0\\
 1 & 1 & 0 & \cdots 0\\
\end{pmatrix}.
$$
The normalization factor doesn't change whether the Jacobian matrix is nonzero or not, and therefore
is interchangeable and omitted here.
At this particular point, the Jacobian matrix is nonzero and proportional to:
\begin{equation}
    [1, 1, 0, 0, 0, 0, 0, 0, 0, 1, 1, 0, 0, 0, 0, 0, 0, 0],
\end{equation}
which suggests that the middle
point is not a local minimum, thus proving the non-convexity of the set of local minimums.
It then follows that the function of $g(C, D)$ is non-convex.

\end{document}